\documentclass[color,12pt]{article}
\usepackage{graphicx}
\usepackage{natbib}
\usepackage{amssymb,amsmath,amsthm}
\usepackage{bbm}
\usepackage{url}
\usepackage{color}
\usepackage{multirow}
\usepackage{hyperref}
\usepackage{mathptmx}      
\usepackage{latexsym}
\usepackage{ifpdf}
\usepackage{algorithm,algorithmic}

\usepackage[section]{placeins}

\newcommand{\argmax}{\operatornamewithlimits{argmax}}

\newcommand{\argu}[2]{\renewcommand{\arraystretch}{0.5} \begin{array}{c} \\
{\displaystyle #1}\\{\scriptscriptstyle #2}\end{array}}
\newcommand{\argmin}{\arg\hspace{-0.1cm}\min}
\newcommand{\II}{{\bf 1}}
\newtheorem{example}{Example}
\newtheorem{definition}{Definition}
\newtheorem{theorem}{Theorem}
\newtheorem{corollary}{Corollary}

\begin{document}
\title{Outlier Detection in Contingency Tables based on Minimal Patterns}
\author{Sonja Kuhnt\footnote{Faculty of Statistics, TU Dortmund University, Germany} \and Fabio Rapallo\footnote{Department DISIT, Universit\`a del Piemonte Orientale, Italy} \and Andr$\acute{\text{e}}$ Rehage$^*$ } 

\date{}

\maketitle

\begin{abstract}
A new technique for the detection of outliers in contingency
tables is introduced. Outliers thereby are unexpected cell counts
with respect to classical loglinear Poisson models. Subsets of
cell counts called minimal patterns are defined, corresponding to
non-singular design matrices and leading to potentially
uncontaminated maximum-likelihood estimates of the model
parameters and thereby the expected cell counts. A criterion to
easily produce minimal patterns in the two-way case under
independence is derived, based on the analysis of the positions of
the chosen cells. A simulation study and a couple of real-data
examples are presented to illustrate the performances of the newly
developed outlier identification algorithm, and to compare it with
other existing methods.



\end{abstract}

\section{Introduction}

In every statistical analysis, observations can occur  which
``appear to be inconsistent with the remainder of that set of
data'' \citep{barnett1994}. The same authors also name outliers in
contingency tables among little-explored areas, which is up to day
still true. For two-way tables outliers have been treated in a
couple of research papers in connection with the multinomial
model, e.g., by employing residuals and by defining suitable tests
based on them in their detection
\citep{simonoff1988,fuchs1980,gupta2007}. Approaches for the
detection of outliers in higher-dimen\-sional tables with respect
to the Poisson model are also found in the literature (e.g.
\citealp{upton1995,kuhnt2004}).

In the context of contingency tables we deal with outlying cells
rather than individual outlying observations contributing to the
cell counts. Therefore, the detection of outliers in contingency
tables is based on a sample of size one for each cell count, and
this fact implies that any detection procedure must be defined
with the greatest caution. Additionally, for more than one
outlying cell, their position in the table can be crucial with
respect to their identification as well as their effect on data
analysis methods. This fact has been recognized in the discussion
of outlier detection methods and breakdown concepts for
contingency tables by \citet{kuhnt2000,kuhnt2010}. Also
\cite{rapallo2012} introduces a notion of patterns of outliers in
connection with goodness-of-fit tests by applying techniques from
algebraic statistics.

In this paper we follow a new approach towards outlier
identification in contingency tables. Going back to the general
notion of outliers as observations (or, more precisely: cells)
deviating from a structure supported by the majority of the data
we define  so-called minimal patterns. These sets cover more than
half of the cells while at the same time containing just enough
cells to ensure full rank of the subdesign matrix of a loglinear
model. For each pattern the remaining cell counts are candidate
outliers. Although the independence model has a major role in the
present paper, our technique based on minimal patterns can be
applied to any loglinear model. Moreover, finding these minimal
patterns is not an easy task, and for the independence model in
two-way tables we derive theoretical results on the nature of
minimal patterns. Thus, the independence model is used as a
leading example. Nevertheless, we discuss also an example on a
three-way table under conditional independence in order to show
the practical applicability of the notion of minimal pattern under
a general loglinear model.

To actually identify the outliers, we suggest two possible
algorithms by running through all minimal patterns and using the
notion of $\alpha$-outliers.

The paper is organized as follows. Section
\ref{sec:BasicDefinitions} briefly recalls $\alpha$-outliers with
respect to  loglinear Poisson models and one-step outlier
identification methods based on ML- or $L_1$-estimators. In
Section \ref{sec:OutlierIdentifiers} we define  (strictly) minimal
patterns and present two outlier detection  methods with minimal
patterns, called OMP and OMPC, the latter identifying cell counts
with the highest count of being an outlier with respect to a
minimal pattern. Some results on the connection between minimal
patterns and cycles in subtables are derived in Section
\ref{sec:MinPatIndependenceModel} for the independence model in
two-way tables. The performances of the different outlier
identification methods are compared by a simulation study in
Section \ref{sec:SimulationStudy} and applications to several data
sets from the literature are discussed in Section
\ref{sec:casestudies}. Finally, in Section \ref{sec:conclusion}
some conclusions and comments are made.

\section{Loglinear Poisson models, estimators and $\alpha$-outliers}
\label{sec:BasicDefinitions}

We consider contingency tables with $N$ cell counts, assumed to be
realizations of random variables $Y_j$, $j=1,...,N$, from a
loglinear Poisson model. Any of these models may be presented as
generalized linear models \citep{agresti2002} with structural
component
\[
 E(Y_j) =  \exp(x_j' \beta)=: m_j  , \ j = 1,...,N
\]
where $x_j$ is the $j^{th}$ column of the full rank design matrix
$X \in \mathbb R^{p \times N}$ of the model  and $\beta\in \mathbb
R^p$ the unknown parameter vector. The maximum likelihood
(ML-)estimator of $\beta$ is given by
\begin{eqnarray}
\label{mlestim} \widehat{\beta}^{ML}
=\argu{\argmax}{\beta\in\mathbb R^p}\left(\sum_{j=1}^N\left( Y_j\,
x_j'\beta-\exp({x_j'\beta})\right)\right).
\end{eqnarray}
A more robust alternative  is  the $L_1$-estimator
\citep{hubert1997}
\begin{eqnarray}
\label{l1estim} \widehat{\beta}^{L_1}
=\argu{\argmin}{\beta\in\mathbb R^p} \sum_{j=1}^N|\log Y_j
-x_j'\beta|.
\end{eqnarray}

Generally, the notion of outliers as surprising observations far
away from the bulk of the data has been formalized by so-called
$\alpha$-outlier regions \citep{davies1993}. Thereby observations
which are located in a region of the sample space with a very
small probability of occurrence with respect to a given model are
defined as outliers. A formal definition of outliers in
contingency tables is given in \cite{kuhnt2004}:

\begin{definition}
An observed cell count $y_j$ is called an $\alpha$-outlier with
respect to a loglinear Poisson model if it lies in the outlier
region
\[
out(\alpha, Poi(m_{j})) =\{y\in\mathbb N: poi(y,m_j)< K(\alpha)\},
\]
where $poi(\cdot,m_j)$ denotes the probability density function of
a Poisson random variable, $\alpha \in (0,1)$, and
$K(\alpha)=\sup\{K>0:\sum_{y\in \mathbb N} poi(y,m_j)
\II_{[0,K]}(poi(y,m_j))\leq \alpha\},$ \mbox{where } $\II_A(x)$ is
the indicator function.
\end{definition}

The complement of an $\alpha$-outlier region is called the
$\alpha$-inlier region. All cell counts within an inlier region
are named inliers, i.e. inliers are just those observations which
are not outliers. Using these notions, one-step outlier
identifiers are easily derived,  defined next based on the
$L_1$-estimator. However,  the estimator type is exchangeable
where robust estimators are of course to be preferred.

\begin{definition} \label{Def:OutlierIdentifier}
Let $\alpha \in (0,1)$ be given. A one-step outlier identifier
based on the $L_1$-estimator is defined by the following
procedure:
\begin{enumerate}
\item[(i)] Estimate $\hat{m}_j, j=1,...,N$, for the loglinear
Poisson model based on the complete contingency table by the
$L_1$-estimator.

\item[(ii)] Identify cell counts $y_j$ in $\alpha$-outlier regions
with respect to $Poi(\hat{m}_j)$ as outliers.
\end{enumerate}
\end{definition}

The choice of $\alpha$ for the one-step outlier identifiers in
relation to the size $N$ of the table is discussed in
\cite{kuhnt2004}. This identifier is compared in Section
\ref{sec:SimulationStudy} to the new methods developed next.

\section{Detecting outliers based on minimal patterns}
\label{sec:OutlierIdentifiers}

Consider the notion of outliers as observations which are
deviating from a model structure supported by the majority of the
data. Here this model is assumed to be a loglinear model
characterized by its design matrix $X$. We look at patterns in the
table, given as subsets of the cells,  which cover at least half
of the table but not more observations than necessary to ensure a
full rank design matrix. These patterns are seen as potential core
sets of  the majority of the data from which individual
observations deviate.

\begin{definition} \label{Def:MinimalPattern}
Let $X$ be the design matrix of a log-linear model with parameter
space ${\mathbb R}^p$. A subset of  cells is called a {\em minimal
pattern} if
\begin{itemize}
\item[(i)] the subset has at least $\left\lfloor \frac{N}{2}
\right\rfloor + 1$ elements;

\item[(ii)] the  corresponding submatrix of $X$ is of full rank;

\item[(iii)] the subset has the minimal number of elements
necessary to fulfill condition (i) and condition (ii).
\end{itemize}
\end{definition}

Restricting the considered subset of the cells to those necessary
to uniquely define model parameters leads to the definition of
strictly minimal patterns.

\begin{definition}
Let $X$ be the design matrix of a log-linear model with parameter
space ${\mathbb R}^p$. A subset of $p$ cells is called a {\em
strictly minimal pattern} if the corresponding submatrix of $X$ is
of full rank.
\end{definition}

 If  $p=\left\lfloor \frac{N}{2} \right\rfloor + 1$ holds, then strictly minimal and minimal patterns coincide.
In case of $p<\left\lfloor \frac{N}{2} \right\rfloor + 1$, adding
$\left\lfloor \frac{N}{2} \right\rfloor + 1-p$ arbitrarily chosen
cells to a strictly minimal pattern returns a minimal pattern.
Note that not all subsets with $p$ cells yield non-singular
matrices.

Strictly minimal patterns are different from strictly
reconstructable replacement patterns  \citep{kuhnt2010}. The
latter define outlier patterns which are unambiguously
identifiable and are used to describe the breakdown behavior of
estimators and identification rules. They are closely related and
inspired by the notion of unconditionally identifiable interaction
patterns in the situation of two-way classification models as
introduced by  \citep{terbeck98}.

Before developing algorithms for the detection of outliers based
on minimal patterns we fix some notation. Let $\mathcal W$ be the
set of all $W$ minimal patterns and $X$ the full design matrix in
the loglinear Poisson model. Each column of $X$ corresponds to a
cell in the contingency table. Taking only the columns of $X$
which correspond to the cells of each minimal pattern yields $X_w,
w=1,...,W$, and we denote with ${\mathcal F}_w$ the set of column
indices in the $w$-th minimal pattern.

A first algorithm on the detection of outliers with minimal
patterns called OMP is defined in Algorithm \ref{alg1} and
identifies the set with the minimal number of elements as
identified outliers. The idea is that we then have the largest
majority of observations well explained by the model given by the
minimal pattern.

\begin{algorithm}
\caption{Outlier detection with minimal patterns (OMP)}
\label{alg1}
\begin{algorithmic}
\FOR{$w=1$ to $W$} \STATE $\widehat{\beta}_w^{ML} \leftarrow
\argu{\argmax}{\beta\in\mathbb R^p}\left(\sum_{1 \leq j \leq N
\cap j \in \mathcal F_w} \left(
Y_j\, x_j'\beta-\exp({x_j'\beta})\right)\right)$ 
\FOR{$j=1$ to $N$} \STATE Determine $out(\alpha,
Poi(\hat{m}^w_j))$ for $m^w_j$ based on
$\exp(x_j'\widehat{\beta}_w^{ML})$
\ENDFOR \STATE $NUMB.OUT_w \leftarrow$ Number of outliers for
minimal pattern $w$ \ENDFOR \FOR{$w=1$ to $W$} \IF{$NUMB.OUT_w =
min(NUMB.OUT)$} \STATE Outlier pattern $\leftarrow $ Cells with
outliers identified with minimal pattern $w$ \ENDIF \ENDFOR
\end{algorithmic}
\end{algorithm}

The  idea behind the new outlier detection methods is to run
through all minimal patterns and consider each of them as
outlier-free subset of the table. The maximum likelihood estimate
from these cells provides estimated mean values for all cells. We
check for all cell counts  outside the pattern if they lie in the
$\alpha$-outlier region with respect to the Poisson distribution
given by the estimate. Those cells for which this is true make the
set of outliers with respect to the  minimal pattern. Hence, we
get a set with outliers for each minimal pattern.

Notice that in  Algorithm \ref{alg1} the minimum number of
outliers may be attained for more than one minimal pattern. Then
more than one solution exist and different possible outlier
patterns are identified, which can be discussed based on knowledge
of the subject.

A slightly different alternative to OMP is implemented in
Algorithm \ref{alg2bis}, called outlier detection with minimal
patterns and the count method (OMPC). Here we count how often each
cell is identified as outlier with respect to a minimal pattern.
If the cell is identified  in more than half of the cases  it is
identified as outlier. We denote the number of minimal patterns
not including the cell by $r$. The choice of the value $r/2$ as a
cut-off in order to discriminate between outliers and inliers is
briefly discussed in Section \ref{sec:SimulationStudy}.

\begin{algorithm}
\caption{Outlier detection with minimal patterns and the count
method (OMPC)} \label{alg2bis}
\begin{algorithmic}
\FOR{$w=1$ to $W$} \STATE $\widehat{\beta}_w^{ML} \leftarrow
\argu{\argmax}{\beta\in\mathbb R^p}\left(\sum_{1 \leq j \leq N
\cap j \in \mathcal F_w} \left(
Y_j\, x_j'\beta-\exp({x_j'\beta})\right)\right)$ 

\FOR{$j=1$ to $N$, $j \notin \mathcal F_w$}
    \STATE Determine $out(\alpha, Poi(\hat{m}^w_j))$ for $m^w_j$ based on          $\exp(x_j'\widehat{\beta}_w^{ML})$
\STATE $r_j \leftarrow $ absolute frequency of cell $j$ not
contained in a minimal pattern \IF{$\# (y_j \in out(\alpha,
Poi(\hat{m}^w_j)), j \notin \mathcal F_w) > r_j/2$ } \STATE $y_j$
is an outlier \ENDIF \ENDFOR  \ENDFOR
\end{algorithmic}
\end{algorithm}

When $W$ becomes large and the enumeration of all minimal patterns
is not feasible, it is possible to introduce a standard Monte
Carlo approximation in the algorithms.

As we take the minimal patterns to be potential outlier-free
subsets it seems straightforward to employ the ML-estimator.
However, the general procedure is open to other choices. Within
the simulation study in Section \ref{sec:SimulationStudy} we also
use the $L_1$-estimator and call the procedure OMPCL1.


\cite{shane2001} also use elemental subsets of the data to derive
robust estimators for categorical data. To detect outliers, we
adapt the Pearson least trimmed chi-squared residuals (LTCS)
estimator developed by them. In short, Shane and Simonoff create a
certain number of elemental subsets and estimate $\beta_{LTCS}$
for each elemental subset. The subset that minimizes the criterion
$\sum_{j=1}^N c_j X^2_{(j)} (y_j, \hat{e}_j)$ will be chosen to
estimate $\beta_{LTCS}$, where \[ c_j = \begin{cases} 1, &
\text{if } j \leq h \\ 0, & \text{if } j > h, \end{cases} \] and
$X^2_{(j)}$ is the ordered Pearson chi-squared statistic. The
authors derive breakdown points for this estimator, which are
based on the tuning parameter $h$. The optimal breakdown point of
$\beta_{LTCS}$ is yielded by \[ h=h_{op} \in [ \left\lfloor (N + G
+ 1)/2 \right\rfloor,  \left\lfloor (N + G + 2)/2 \right\rfloor ],
\] where $G$ is the maximum number of linearly independent rows in
the design matrix $X$. The generation of elemental subsets can be
conducted through minimal patterns, since the authors choose
subsets with $p$ elements and their design matrices having full
rank. They also state that the number of elemental subsets does
not have to be very large for their method. In their simulation
study, results based on $500$ elemental subsets were virtually the
same than for $1,000$ or $1,500$, hence we choose $W=1,000$. The
algorithm is summarized in the following pseudo-code:

\begin{algorithm}
\caption{Outlier detection with minimal patterns and the LTCS
estimator (OLTCS)} \label{alg3ltcs}
\begin{algorithmic}
\STATE $h \leftarrow \left\lfloor (N + G + 2)/2 \right\rfloor$
\STATE $c_j \leftarrow \II_{\{ j \leq h\}}(j), j = 1,...,N.$
\FOR{$w=1$ to $1000$} \STATE $\widehat{\beta}_w^{ML} \leftarrow
\argu{\argmax}{\beta\in\mathbb R^p}\left(\sum_{1 \leq j \leq N
\cap j \in \mathcal F_w} \left(
Y_j\, x_j'\beta-\exp({x_j'\beta})\right)\right)$ 
\STATE $\hat{e}^w \leftarrow y - \exp(X\hat{\beta}^{ML}_{w})$
\STATE $X^2_w \leftarrow \sum_{1 \leq j \leq N \cap j \in \mathcal
F_w} (y_j - \hat{e}^w_j)^2 / \hat{e}^w_j$
\ENDFOR \STATE $w^* \leftarrow \argu{\argmin}{w} \sum_{j=1}^N c_j
X^2_{w; (j)} (y_j, \hat{e}_j )$ \FOR{$j = 1$ to $N$} \STATE
Determine $out(\alpha, Poi(\hat{m}^{w^*}_j))$ for $m^{w^*}_j$
based on          $\exp(x_j'\widehat{\beta}_{w^*}^{ML})$
\IF{$\# (y_j \in out(\alpha, Poi(\hat{m}^{w^*}_j)))$ } \STATE
$y_j$ is an outlier \ENDIF \ENDFOR
\end{algorithmic}
\end{algorithm}

Note that Algorithm \ref{alg3ltcs} can be slightly adjusted to
perform outlier detection based on minimal patterns and the LMCS
estimator (see \citealp{shane2001}) as well by replacing $c_j$
with $c_h = 1$ and $c_j = 0, j \neq h$. However, simulation
results suggest that there is virtually no difference between
outlier detection techniques based on LMCS or LTCS, hence we
restrict the analyses to the LTCS estimator.


\section{Minimal patterns and cycles in  the independence model}
\label{sec:MinPatIndependenceModel} Running through all possible
subsets of dimension $p$ to determine the strictly minimal
patterns quickly becomes unfeasible for larger dimensional tables.
It is therefore important to  analyze the structure of these
patterns in more detail.

We focus on the loglinear independence model for two-dimensional
$I \times J$ contingency tables, assuming without loss of
generality $I \leq J$. The design matrix $X$  can be expressed as:
\begin{equation}
X = [a_0, r_1 , \ldots, r_{I-1}, c_1, \ldots, c_{J-1} ]' \, ,
\end{equation}
where  $a_0$ is a unit vector, $r_1$ is the indicator vector of
the first row, $c_1$ is the indicator vector of the first column,
and so on. For instance, the design matrix for $3\times3$ tables
is:
\begin{eqnarray} X &=&
 \begin{pmatrix}
      1 &   1  &  1  &  1   & 1 &   1 &   1 &  1  &  1 \\
      1 &   1  &  1  &  0   & 0 &   0 &   0 &  0  &  0 \\
      0 &   0  &  0  &  1   & 1 &   1 &   0 &  0  & 0 \\
      1 &   0  &  0  &  1   & 0 &   0 &   1 &  0  & 0\\
      0 &   1  &  0  &  0   & 1 &   0 &   0 &  1  & 0 \\
\end{pmatrix} \, . \label{Xdes3x3-theo}
\end{eqnarray}
Another classical representation of the same model is given by the
design matrix
\begin{eqnarray} \widetilde X &=&
 \begin{pmatrix}
      1 &   1  &  1  &  1   & 1 &   1 &   1 &  1  &  1 \\
      1 &   1  &  1  &  0   & 0 &   0 &  -1 &  -1  & -1 \\
      0 &   0  &  0  &  1   & 1 &   1 &  -1 &  -1  & -1 \\
      1 &   0  & -1  &  1   & 0 &  -1 &   1 &  0  & -1 \\
      0 &   1  & -1  &  0   & 1 &  -1 &   0 &  1  & -1 \\
\end{pmatrix} =: \widetilde X_{3 \times 3}. \label{Xdes3x3}
\end{eqnarray}
We  use the latter parametrization in our simulation study as it
is the usual parametrization implemented in the software for
loglinear models, while we use the former parametrization in the
proofs, as many formulae become easy to handle.

In this model, the relevant parameter space for the unknown
parameter vector $\beta$ is ${\mathbb R}^{(I+J-1)}$. Table
\ref{tab:numberofminimalpatterns} shows that the number of
possible patterns with $p=I+J-1$ cells  as well as the number of
(strictly) minimal patterns increases quickly for higher
dimensional tables.

\begin{table}[h]
\centering
\begin{tabular}{rrrrrrrr}
\hline
Dimension of the table                                           & $3 \times 3$  & $2 \times 5$  & $3 \times 4$  & $3 \times 5$  & $4 \times 4$  & $3 \times 6$  & $4 \times 5$ \\
\hline
$p=I+J-1    $                               & 5                         & 6                         & 6                         & 7                         & 7                          & 8                        & 8 \\
$\phi = \left\lfloor \frac{N}{2} \right\rfloor   + 1$                               & 5                         & 6                         & 7                          & 8                        & 9                          & 10                        & 11 \\
$\binom{N}{\phi}$& 126                  & 210                   & 792                   & 6435                  & 11440                 & 43758                  & 167960 \\
$W = \#$  min. patterns             & 81                        & 80                        & 612                       & 3780                  & 9552                   & 26325                    & 139660 \\
$\#$ str. min. patterns & 81 & 80 & 432 & 2025 & 4096 & 41066 & 105408 \\
\hline
\end{tabular}
\caption{Number of minimal patterns for different independence
models} \label{tab:numberofminimalpatterns}
\end{table}

\begin{example}
In the case of $3 \times 3$ tables, the two configurations below
have different behavior:
\begin{center}
\begin{tabular}{|c|c|c|}
\hline
$\star$ & $\star$ &   \\
\hline
$\star$ & $\star$ &   \\
\hline
 & & $\star$ \\
\hline
\end{tabular}
\qquad \qquad \qquad
\begin{tabular}{|c|c|c|}
\hline
$\star$ & $\star$ &  \\
\hline
 & $\star$ & $\star$ \\
 \hline
 & & $\star$ \\
\hline
\end{tabular}
\end{center}
(the $\star$'s denote the chosen cells). The configuration on the
left hand side produces a singular submatrix, while the
configuration on the right hand side produces a non-singular
matrix, and hence it is a strictly minimal pattern. At a first
glance, we note that in the singular case there is a complete $2
\times 2$ subtable among the chosen cells, while in the other case
it is not. The relevance of $2 \times 2$ subtables in the study of
the independence model is well known, see e.g. \cite{agresti2002},
and a different perspective within the field of Algebraic
Statistics is investigated in e.g. \cite{rapallo2003}. However,
the simple notion of a $2 \times 2$ submatrix is not sufficient to
effectively describe the problem, as shown in the following
example:
\begin{center}
\begin{tabular}{|c|c|c|c|}
\hline
$\star$ & $\star$ &  &  \\
\hline
$\star$ & & $\star$ & \\
\hline
 & $\star$ & $\star$ &  \\
\hline
 & &  & $\star$ \\
 \hline
\end{tabular}
\end{center}
In this case, the chosen configuration does not contain any $2
\times 2$ submatrices, and nevertheless the corresponding
submatrix is singular.
\end{example}

To explore the structure of patterns in the table we need to
introduce the notion of $k$-cycle.

\begin{definition}\label{Def:cycle}
Let $k \geq 2$. A $k$-cycle is a set of $2k$ cells contained in a
$k \times k$ subtable, with exactly $2$ cells in each row and in
each column of the submatrix.
\end{definition}

\begin{example}
In view of Definition \ref{Def:cycle}, a $2$-cycle is simply a $2
\times 2$ submatrix, while a $3$-cycle is a set of $6$ cells of
the form
\begin{center}
\begin{tabular}{|c|c|c|}
\hline
$\star$ & $\star$ &  \\
\hline
$\star$ &  & $\star$ \\
 \hline
 & $\star$ & $\star$ \\
\hline
\end{tabular}
\end{center}
\end{example}

In case of the independence model, the following theorem shows
that the cycles are the key ingredient to check whether a subset
of $p$ cells is a strictly minimal pattern.

\begin{theorem}
A set of $p=I+J -1$ cells forms a strictly minimal pattern for the
independence model if and only if it does not contain any
$k$-cycles, $k=2, \ldots, I$.
\end{theorem}

\begin{proof}
First, note that a cycle can be decomposed into two subsets of $k$
cells each with one cell in each row and in each column. It is
enough to sum the columns of the design matrix $X$ with
coefficient $+1$ for the cells in the first subset and with
coefficient $-1$ for the second subset and we obtain a null
vector. Thus, the submatrix is singular and the set does not form
a strictly minimal pattern.

Conversely, if the submatrix is singular, then there is a null
linear combination among the columns of the submatrix, with
coefficients not all zero. Denote with $c_{(i,j)}$ the column of
the design matrix corresponding to the cell $(i,j)$. Therefore, we
have
\begin{equation} \label{lincomb}
\gamma_1 c_{(i_1,j_1)} + \ldots + \gamma_p c_{(i_p,j_p)} = 0
\end{equation}
and the coefficients $\gamma_1, \ldots, \gamma_p$ are not all
zero. Without loss of generality, suppose that $\gamma_1 >0$. As
the indicator vector of row $i_1$ belongs to the row span of $X$
and the same holds for the indicator vector of column $j_1$, we
must have: a cell in the same row $(i_2,j_2)=(i_1,j_2)$ with
negative coefficient in Eq. \eqref{lincomb}; a cell in the same
column $(i_3,j_3)=(i_3,j_1)$ with negative coefficient in Eq.
\eqref{lincomb}. Therefore, there must be a cell in row $i_3$ and
a cell in column $j_2$ with positive coefficients. Now, two cases
can happen:
\begin{itemize}
\item if the cell $(i_3,j_2)$ is a chosen cell and its coefficient
in Eq. \eqref{lincomb} is positive, we have a $2$-cycle;

\item otherwise, we iterate the same reasoning as above, with
another pair of cells.
\end{itemize}
This shows that there exists a certain number $k$ of rows ($k>2$),
and the same number of columns, with two cells each with a
non-zero coefficient. Such cells form by definition a $k$-cycle.
\end{proof}

As a corollary, the following algorithm produces strictly minimal
patterns:
\begin{enumerate}
\item Let ${\mathcal C}$ be the set of all cells of the table, and
${\mathcal S}=\emptyset$ the set of the chosen cells.

\item For $q \in \{1, \ldots, I+J - 1\}$:
\begin{itemize}
\item Choose a cell uniformly from ${\mathcal C}$, add it to
${\mathcal S}$, and delete it from ${\mathcal C}$;

\item Find all $3$-tuples, $5$-tuples and so on of cells in
${\mathcal S}$ containing the chosen cell and delete from
${\mathcal C}$ all cells (if any) producing $2$-cycles, $3$-cycles
and so on.
\end{itemize}
\end{enumerate}

Notice that the first three cells are chosen without any
restrictions. Moreover, as the algorithm is symmetric on row and
column permutation, one has that the strictly minimal pattern is
selected uniformly in the set of all strictly minimal patterns.

For $3 \times 3$ tables, our statement is equivalent to another
criterion, to be found in \cite{kuhnt2000}.

\begin{corollary}
For the independence model for $3 \times 3$ tables, the absence of
$2$-cycles is equivalent to:
\begin{itemize}
\item[(i)] no empty rows;

\item[(ii)] no empty columns;

\item[(iii)] for each selected cell, there is at least another
cell in the same row or in the same column.
\end{itemize}
\end{corollary}

\begin{proof}
Suppose that there is an empty row. In the remaining two rows we
have to put $5$ cells, and a $2$-cycle must appear. The same
reasoning holds in the case of an empty column. Finally, if there
is a selected cell, say $(i,j)$, with no other cells in the same
row or in the same column, we exclude for the remaining 4 cells of
the minimal pattern the 5 cells of the $i$-th row and of the
$j$-column. Thus the remaining $4$ cells are forced to constitute
a $2$-cycle.

On the other hand, suppose that there is a $2$-cycle, and suppose
without loss of generality that the cycle is formed by the cells
$(1,1),(1,2),(2,1),(2,2)$. The last selected cell can be chosen in
$5$ different ways. In two cases, $(1,3)$ or $(2,3)$, we have an
empty row; in two cases, $(3,1)$ or $(3,2)$, we have an empty
column; in the last case, $(3,3)$, this cell has no other cells in
the same row or in the same column.
\end{proof}

For a general loglinear model, we can define an algorithm to
efficiently sample minimal patterns as follows:

\begin{itemize}
\item[(a)] First, choose a strictly minimal pattern.

\item[(b)] Add randomly chosen cells in order to achieve a minimal
pattern, if needed.
\end{itemize}

This procedure can be repeated until every possible minimal
pattern has been found. Alternatively, if this is unfeasible due
to the dimension of the table, the procedure may be stopped after
a certain time or certain number of patterns. In case of the
two-way independence model this  produces a uniform random minimal
pattern, as long as the strictly minimal pattern is uniformly
chosen with the algorithm above.


\section{Simulation study}
\label{sec:SimulationStudy}

In the previous sections we presented different methods to
identify $\alpha$-outliers. To compare different outlier
identifiers, \cite{kuhnt2010} discusses breakdown points of the
methods. For the OMPC and the OMP methods, it is not clear if
breakdown points or similar criteria can be derived theoretically
at all. Hence we present three loglinear Poisson models with
varying outlier situations, conduct simulations and check whether
the methods (one-step $L_1$ (OL1), OLTCS, OMPC and OMPCL1) detect
outliers and inliers correctly. We exclude OMP from the comparison
as it might lead to results which are not unique and therefore not
directly comparable.

We consider three different loglinear Poisson models (($3 \times
3$), ($4 \times 4$) and ($10 \times 10$)) and insert various
outlying values in the simulated contingency tables. For example,
we vary the $\alpha$-value which determines the outlyingness of
the inserted value. For the simulations we adapt the notion of
``types'' and ``antitypes'' from Configural Frequency Analysis
\citep{voneye:02}. A type is defined as a cell in a contingency
table with a higher value than the expected cell count, hence
above the upper bound of the corresponding $\alpha$-inlier region.
An antitype has a smaller value than the expected cell count,
hence smaller than the lower bound of the corresponding
$\alpha$-inlier region.

The six simulated scenarios are described below. The simulations
were performed with R \citep{r-project:10} and the results are
given in Table \ref{tab:sim}.

\begin{enumerate}
\item We generate 100 $3 \times 3$ contingency tables with
$\widetilde X = \widetilde X_{3 \times 3}$ and $\beta_1 = (4, 0.2,
-0.2, 0.4, 0.3)'$ with only one $\alpha$-outlier
($\alpha=10^{-4}$) in cell (1,1). Since the position of one
outlier in the table is unimportant we place the outlier in the
first row and column of each table. The outlier can be seen as a
moderate outlier. For the cell $(1,1)$, the outlier region with
respect to a Poisson distribution is given by:
\[ [0, out_{\rm{left}})
\cup (out_{\rm{right}}, \infty) = [0, 63) \cup (140, \infty)
\]
such that the value 62 is inserted as antitype and 141 as type.
Since $3 \times 3$ contingency tables have been analyzed in
\cite{kuhnt2000} extensively, we then move to larger tables.

\item  We generate
 100 $4 \times 4$ tables based on $\widetilde{X}_{4\times4}$ created analogous to $\widetilde{X}_{3\times3}$ and\\  $\beta_2 = (3.8,
0.2,-0.2,0.1,0.25,0.3,-0.1)'$. As before, we insert only one
moderate $\alpha$-outlier $(\alpha=10^{-4})$ in cell (1,1), namely
$n_{11}=39$ as antitype and $n_{11}=105$ as type.

\item Again we generate 100 $4 \times 4$ tables based on $\beta_2$
and $\widetilde{X}_{4\times4}$. To see how the methods work with
several outliers, we add another $\alpha$-outlier
$(\alpha=10^{-4})$ resulting in three different situations: Two
types, two antitypes, and one type and one antitype. In this
scenario, we inserted the outliers in cells $(1,1)$ and $(1,2)$.
Notice that the presence of two outliers in the same row can
manipulate the estimates of that row in a notably way.


\item We reconsider the situation from the third scenario with
$\beta_2$ and $\widetilde{X}_{4\times4}$. This time, we replace
two values on the main diagonal of the contingency table with
$\alpha$-out\-liers $(\alpha=10^{-4})$ in cells $(1,1)$ and
$(2,2)$. In this case, the two outliers affect different parameter
estimates.


\item The last simulation with $4 \times 4$ tables based on
$\beta_2$ and $\widetilde{X}_{4\times4}$ is similar to the third
scenario, but here the outlyingness of the replaced values in
cells (1,1) and (1,2) is more extreme. Now, $\alpha=10^{-8}$ is
used.

\item \label{sec:Sim6} We finish the simulation studies with the
generation of $100$ large $10 \times 10$ contingency tables. The
corresponding parameter vector is
\begin{eqnarray*}
\beta_3 & = & (3.3, 0.2,-0.2,0.1,0.25,0.3,-0.1,0.4,0.2,0.1, \\
      &   &  0.2,-0.4, 0.2,-0.2,0.1,0.0,0.1,-0.3,0.1)' \,
\end{eqnarray*}

and the design matrix given by $\widetilde{X}_{10\times10}$. Then
$\alpha$-outliers are inserted in cell $(1,1)$ and cell $(2,3)$,
with $\alpha=10^{-4}$. The number of minimal patterns we consider
here is constrained to $500$.

\end{enumerate}

\begin{table}[hp]
\centering
\begin{tabular}{c|c}
Scenario &  \\ \hline 1 &
\begin{tabular}{rrrrr}
 & \multicolumn{2}{c}{$n_{11} = 62$} & \multicolumn{2}{c}{$n_{11} = 141$}  \\

estimator & outliers & inliers & outliers & inliers \\ 
\hline
OL1 & 0.320     & 0.963     & 0.480 & 0.974 \\ 
OLTCS & 0.480           & 0.946         & 0.530 & 0.950 \\
OMPC & 0.680    & 0.754 & 0.820 & 0.773 \\ 
\end{tabular} \\ \hline
2 &
\begin{tabular}{rrrrr}
& \multicolumn{2}{c}{$n_{11} = 39$} & \multicolumn{2}{c}{$n_{11} = 105$} \\
estimator & outliers & inliers & outliers & inliers \\ 
\hline
OL1 & 0.620 & 0.979  & 0.620 & 0.989 \\ 
OLTCS & 0.740 & 0.943 & 0.680 & 0.937 \\
OMPC & 0.890 & 0.899  & 0.900 & 0.909 \\
OMPCL1 & 0.910 & 0.877 & 0.930 & 0.909\\
\end{tabular} \\ \hline
3 &
\begin{tabular}{rrrrrrr}
& \multicolumn{2}{c}{2 antitypes} & \multicolumn{2}{c}{1 type, 1 antitype} & \multicolumn{2}{c}{2 types} \\
& \multicolumn{2}{c}{$n_{11} = 39, n_{12} = 42$} & \multicolumn{2}{c}{$n_{11} = 39, n_{12} = 110$} & \multicolumn{2}{c}{$n_{11} = 105, n_{12} = 110$} \\

estimator & outliers & inliers & outliers & inliers & outliers & inliers \\
\hline
OL1 & 0.035 & 0.960 & 0.725 & 0.986 & 0.200 & 0.983 \\
OLTCS & 0.295 & 0.896 & 0.680 & 0.937 & 0.295 & 0.908 \\
OMPC & 0.435 & 0.868 & 1.000 & 0.878 & 0.470 & 0.901 \\
OMPCL1 & 0.615 & 0.833 & 1.000 & 0.846 & 0.720 & 0.874 \\
\end{tabular} \\ \hline
4 &
\begin{tabular}{rrrrrrr}
& \multicolumn{2}{c}{2 antitypes} & \multicolumn{2}{c}{1 type, 1 antitype} & \multicolumn{2}{c}{2 types} \\
& \multicolumn{2}{c}{$n_{11} = 39, n_{22} = 23$} & \multicolumn{2}{c}{$n_{11} = 39, n_{22} = 79$} & \multicolumn{2}{c}{$n_{11} = 105, n_{22} = 79$} \\
Estimator & outliers & inliers & outliers & inliers & outliers & inliers \\
\hline
OL1 & 0.740 & 0.976 & 0.495 & 0.980 & 0.635 & 0.984 \\
OLTCS & 0.795 & 0.948 & 0.670 & 0.923 & 0.695 & 0.921 \\
OMPC & 0.975 & 0.804 & 0.840 & 0.834 & 0.965 & 0.857 \\ 
OMPCL1 & 0.975 & 0.692 & 0.885 & 0.776 & 0.975 & 0.818 \\
\end{tabular} \\ \hline
5 &
\begin{tabular}{rrrrrrr}
& \multicolumn{2}{c}{2 antitypes} & \multicolumn{2}{c}{1 type, 1 antitype} & \multicolumn{2}{c}{2 types} \\
& \multicolumn{2}{c}{$n_{11} = 27, n_{12} = 29$} & \multicolumn{2}{c}{$n_{11} = 27, n_{12} = 128$} & \multicolumn{2}{c}{$n_{11} = 124, n_{12} = 128$} \\
estimator & outliers & inliers & outliers & inliers & outliers & inliers \\
\hline
OL1 & 0.140 & 0.896 & 0.980 & 0.987 & 0.450 & 0.969 \\
OLTCS & 0.370 & 0.870 & 0.835 & 0.936 & 0.430 & 0.898 \\
OMPC & 0.880 & 0.771 & 1.000 & 0.643 & 0.855 & 0.829 \\
OMPCL1 & 0.975 & 0.725 & 1.000 & 0.653 & 0.950 & 0.806 \\
\end{tabular} \\ \hline
6 &
\begin{tabular}{rrrrrrr}
& \multicolumn{2}{c}{2 antitypes} & \multicolumn{2}{c}{1 type, 1 antitype} & \multicolumn{2}{c}{2 types} \\
& \multicolumn{2}{c}{$n_{11} = 18, n_{23} = 9$} & \multicolumn{2}{c}{$n_{11} = 18, n_{23} = 49$} & \multicolumn{2}{c}{$n_{11} = 67, n_{23} = 49$} \\
estimator & outliers & inliers & outliers & inliers & outliers & inliers \\
\hline
OL1 & 0.963 & 0.991 & 0.936 & 0.992 & 0.935 & 0.992 \\
OLTCS & 0.850 & 0.929 & 0.855 & 0.933 & 0.815 & 0.932 \\
OMPC & 0.990 & 0.940 & 0.990 & 0.953 & 1.000 & 0.956 \\
OMPCL1 & 0.910 & 0.995 & 0.895 & 0.997 & 0.940 & 0.997\\
\end{tabular} \\ \hline
\end{tabular}
\caption{Proportions of correctly classified outliers and inliers
in the $6$ simulation scenarios.} \label{tab:sim}
\end{table}

\begin{table}[hp]
\centering
\begin{tabular}{c|c|ccccccccc}
 & & $g=$0.1 & 0.2 & 0.3 & 0.4 & 0.5 & 0.6 & 0.7 & 0.8 & 0.9
\\ \hline
$3 \times 3$ & $M_0$ & 0.984 & 0.976 & 0.950 & 0.898 & 0.828 &
0.752 & 0.650 & 0.498 & 0.450 \\
 & $M_1$ & 0.452 & 0.508 & 0.664 & 0.778 & 0.888 & 0.932 & 0.964 & 0.980 &
 0.984 \\ \hline
$4 \times 4$ & $M_0$ & 0.998 & 0.988 & 0.984 & 0.956 & 0.894 &
0.844 & 0.746 & 0.608 & 0.418 \\
 & $M_1$ & 0.562 & 0.730 & 0.814 & 0.886 & 0.926 & 0.956 & 0.980 & 0.990 &
 0.996 \\ \hline
$5 \times 5$ & $M_0$ & 1.000 & 0.998 & 0.990 & 0.964 & 0.926 &
0.868 & 0.768 & 0.644 & 0.462 \\
 & $M_1$ & 0.632 & 0.792 & 0.892 & 0.928 & 0.950 & 0.980 & 0.992 & 0.996 &
 1.000 \\ \hline
$6 \times 6$ & $M_0$ & 1.000 & 1.000 & 0.990 & 0.984 & 0.956 &
0.936 & 0.880 & 0.786 & 0.588 \\
 & $M_1$ & 0.746 & 0.862 & 0.914 & 0.942 & 0.962 & 0.974 & 0.982 & 0.996 & 0.998
 \\ \hline
$7 \times 7$ & $M_0$ & 1.000 & 1.000 & 1.000 & 0.994 & 0.980 &
0.952 & 0.912 & 0.814 &  0.654 \\
 &$M_1$ & 0.764 & 0.868 & 0.908 & 0.954 & 0.964 & 0.972 & 0.980 & 0.994 &
 0.998 \\ \hline
\end{tabular}
\caption{Proportions of correct classification of cell $(1,1)$
under the models $M_0$ (where (1,1) is not an outlier) and $M_1$
(where (1,1) is an outlier)} \label{tab:coh-2}
\end{table}

All outlier identification methods are always calculated with
$0.01$-outlier regions of the model given by the parameter
estimate. We judge the different methods by the proportion of
correctly identified outliers and inliers, given in Table
\ref{tab:sim}. Analyzing the results, some comments are now in
order.

\begin{itemize}
\item Scenarios 1 and 2 show that the behavior of OL1 is not
satisfactory for small tables. The OLTCS procedure has a better
sensitivity to outliers while only few inliers are classified
wrongly. On the other hand, OMPC has a proportion of correctly
classified outliers notably higher than these methods. OMPCL1 is
not listed in Scenario 1 because of the exact-fit property of the
ML- and $L_1$-estimator in $(3 \times 3)$ tables for minimal
patterns, hence both procedure produce the same result.

\item Scenarios 3 and 4 prove that the position of the outlying
cells within the table is a major issue. In fact, placing the two
outliers in the same row, the proportion of correctly classified
outliers reduces considerably. This phenomenon is particularly
evident in case of two types or two antitypes, since in such cases
the outliers give rise to relevant changes in the parameter
estimates. With two antitypes in the same row we find again that
the OL1 method is almost futile. With respect to Scenario 3, it
seems that the OMPC method outperforms the OLTCS method concerning
outliers in equal directions. On the other hand the inlier
detection rate of the OMPC and OMPCL1 method is notably smaller if
the outliers are in different rows.

\item Comparing scenarios 3 and 5, we observe that all procedures
perform better in finding outliers when the outlyingness of the
two cells is higher.

\item Scenario 6 shows that the proposed methods are still valid
for larger tables, even though the differences between the three
methods become less relevant. The OMPC outperforms the OLTCS,
particularly with regard to outlier detection. Furthermore this is
the only scenario where the OMPCL1 has always a smaller outlier
detection rate and a higher inlier detection rate than the OMPC
method.

\item In all simulations the OMPC algorithm is slightly less
efficient in detecting inliers. This means that in some few cases
it finds more outliers than expected. This issue will be discussed
again after the real data examples, in connections with the
behavior of the OMP method.
\end{itemize}

Finally, we need to motivate our choice of $r_j/2$ as the cutoff
value in the OMPC algorithm. We considered a simulation study for
$3 \times 3$, $4 \times 4$, $5 \times 5$, $6 \times 6$, and $7
\times 7$ tables. For each case, we generated $1,000$ random
contingency tables under two different models:
\begin{itemize}
\item $M_0$: the null independence model;

\item $M_1$: the model of independence plus a $10^{-4}$-outlier in
the cell $(1,1)$.
\end{itemize}
For each table, the $\beta$ vector was chosen with random uniform
components on $(-0.5, 0.5)$ except from the first component, fixed
at $3.8$ in order to control the mean sample size.

Then, we computed the proportion of correct classification of the
cell $(1,1)$ under the two models (i.e., the proportion of outlier
detected for $M_1$ and the proportion of outlier not detected for
$M_0$) running the OMPC algorithm with $\alpha = 0.01$. In order
to motivate the choice of the cutoff point, we have computed such
proportions for different cutoffs of the form $gr$ $(0 < g < 1)$.

The results are displayed in Table \ref{tab:coh-2}. The values in
Table \ref{tab:coh-2} show that $g=1/2$ is a reasonable choice, as
it represents the best trade-off between the two types of error.

\section{Case studies}
\label{sec:casestudies}

We next apply our new outlier detection methods to data sets from
the literature. The first data set of artifacts discovered in
Nevada is a typical example for outliers from the independence
model. The second example of the social mobility of fathers and
sons is widely treated in the literature with the general
understanding that it actually requires a quasi-independence
model, which should become apparent within detecting outliers. The
last example on social networks goes beyond two-way tables and
shows the applicability of our methods to general loglinear
models.
\subsection{Artifacts discovered in Nevada}

To see how the various outlier identification algorithms work
compared to
 procedures from the literature, we look at the data in Table \ref{tab:Nevada}
\citep{mosteller1985}. This table shows how far away from
permanent water certain types of archaeological artifacts have
been found.

\begin{table}[h]
\centering
\begin{tabular}{ccccc}
\hline
                                                                   & \multicolumn{4}{c}{\textit{Distance from permanent water}} \\
\textit{Artifact}          &   Conti-   & Within        & 0.25 - 0.5  & 0.5 - 1 \\
\textit{type}                 & guity      & 0.25 mi    & mi             & mi \\
\hline
Drills                & 2                 & 10             & 4             & 2 \\
Pots                  & 3                 & 8              & 4             & 6 \\
Grinding stones & 13             & 5                  & 3                   & 9 \\
Point fragments & 20                  & 36                  & 19              & 20 \\
\hline
\end{tabular}
\caption{Archaeological finds discovered in Nevada, from
\cite{mosteller1985}. } \label{tab:Nevada}
\end{table}


The OL1 method yields no outliers for $\alpha=0.001$. This holds
also for the OMP method. In contrast, the OMPC method finds two
outliers for $\alpha=0.001$, i.e. cells $(3,1)$ and $(3,2)$. The
OLTCS method detects cell $(3,1)$ as outlier, which is also valid
for the OMPCL1 method. Looking at the OMPC method with a smaller
$\alpha = 0.0005$ we find that only cell $(3,1)$ stays an outlying
cell in this method. This dataset has also been studied in
\cite{simonoff1988}, where cell $(3,1)$ has been declared as
``sure outlier'' and cell $(3,2)$ can be seen as a border-line
situation.

\subsection{Study of social mobility in Britain}

As second example, we briefly present the results on an example
dataset from \cite{glass1954}. The status categories of fathers
and their sons are put together in a (7 $\times$ 7)-contingency
table. \cite{goodman1971} merges certain classes which yields the
$3 \times 3$ contingency table in Table \ref{tab:Status}.
\begin{table}[h]
\centering
\begin{tabular}{rrrrr}
\hline
       &        & \multicolumn{3}{c}{Son} \\
       &        & high & middle & low \\
\hline
       & high   & 588  & 395    & 159 \\
Father & middle & 349  & 714    & 447 \\
       & low    & 111  & 320    & 411 \\
\hline
\end{tabular}
\caption{Status categories of fathers and sons from
\cite{glass1954}.} \label{tab:Status}
\end{table}

Here, OMP identifies the observations $n_{11}, n_{22}, n_{33}$ as
outliers. The OMPC as well as the OMPCL1 method identify every
cell as an outlier, which seems surprising on the one hand, but on
the other hand it is coherent since the underlying independence
model is obviously the wrong one. The choice of the model seems to
be more important to the OMPC and OMPCL1 methods than to the
others. The OMP method yields the only intuitively plausible
outlier pattern with the main diagonal. A potential alternative is
given by the OL1 and OLTCS methods ($n_{11}, n_{13},
n_{31}$ and $n_{33}$ are outliers), while the OMPC and OMPCL1 
offer no satisfying results in this case.

\subsection{Social networks}  \label{beisp2}

As a final example we consider a model different from
independence. \cite{mckinley1973} present a study concerning lay
consultation and help-seeking behavior based on eighty-seven
working-class families in Aberdeen. We consider a
three-dimensional table on friendship networks of pregnant woman
from this data set. The first variable concerns the frequency of
interactions with friends, measured as daily ($X_1=1$), once a
week or more ($X_1=2$) and less than once a week ($X_1=3$). The
geographic proximity to the friends is covered by variable $X_2$
with the categories walk ($X_2=1$) and bus ($X_2=2$). The last
variable $X_3$ states whether the woman is pregnant with the first
($X_3=2$) or a further child ($X_3=1$). The data are summarized in
Table \ref{tab:socialnetworks}.

\begin{table}[h]
\begin{center}
\begin{tabular}{cccccc} \hline
& &  \multicolumn{4}{c}{$X_2$: \textit{Distance}} \\
&  & \multicolumn{2}{c}{Walk} & \multicolumn{2}{c}{Bus}\\
  \multicolumn{2}{c}{$X_3$: \textit{Parity}} & Not first & First & Not first & First \\
\hline
\multirow{3}{8mm}{$X_1$: \textit{Freq.}} & Daily  & $30$ & $6$ & $2$ & $13$\\
 & Weekly & $19$ & $12$ & $16$ & $8$\\
 & Less & $5$ & $2$ & $10$ & $4$ \\ \hline
\end{tabular}
\end{center}
\caption{Data set on social networks from \cite{mckinley1973}.}
\label{tab:socialnetworks}
\end{table}
%

The model we consider assumes the conditional dependence between
$X_1$ and $X_3$ given $X_2$ and has design matrix
\begin{equation*} \widetilde X =
 \left( {\begin{array}{cccccccccccc}
 1 & 1 & 1 & 1 & 1 & 1 & 1 & 1 & 1 & 1 & 1 & 1 \\
 1 & 1 & 1 & 1 & 0 & 0 & 0 & 0 &-1 &-1 &-1 &-1 \\
 0 & 0 & 0 & 0 & 1 & 1 & 1 & 1 &-1 &-1 &-1 &-1 \\
 1 & 1 &-1 &-1 & 1 & 1 &-1 &-1 & 1 & 1 &-1 &-1 \\
 1 &-1 & 1 &-1 & 1 &-1 & 1 &-1 & 1 &-1 & 1 &-1 \\
 1 & 1 &-1 &-1 & 0 & 0 & 0 & 0 &-1 &-1 & 1 & 1 \\
 0 & 0 & 0 & 0 & 1 & 1 &-1 &-1 &-1 &-1 & 1 & 1 \\
 1 &-1 &-1 & 1 & 1 &-1 &-1 & 1 & 1 &-1 &-1 & 1 \\
\end{array} } \right) \, .
\end{equation*}
Running the five outlier identification methods, we obtain that
with the OL1 method the two extreme values are classified as
outliers: $n_{111} = 30$ and $n_{121} = 2$. The OLTCS method
yields one outlier in cell $n_{121}$.

Now we compare the previous results with those yielded by minimal
patterns. There are $\binom{12}{8} = 495$ sets with eight elements
each and $144$ of them fulfill Definition
\ref{Def:MinimalPattern}. The minimal patterns yield 40 times
three outliers, 88 times two outliers and 16 times one outlier.
Therefore we look at those cases where the OMP method found only
one outlier, more precisely cell $n_{121}$ and cell $n_{122}$
(eight times each). So, this method yields two different
solutions.

The OMPC method produces similar results. A cell can be detected
as an outlier 48 times at most. The cells $n_{111}, n_{121},
n_{122}$ have been detected 48 times, cells $n_{311}$ and
$n_{312}$ have not been detected as outliers, the rest of the
cells have been found 24 times, hence $50\%$ of the possible
cases. It is conspicuous that a cell is either always an outlier,
in $50\%$ of the cases or not at all. This fact holds also for
other cell counts and the given model. Furthermore, we are not
interested in having 10 outliers and 2 inliers, that's why we
declare only the cells $n_{111}, n_{121}, n_{122}$ as outliers.
The OMPCL1 method yields the outliers detected by OMPC plus two
additional outliers in cells $n_{112}$ and $n_{311}$. The
comparison of the results from the four methods are summarized in
Table \ref{tab:SocialNetwork}.

\begin{table}[hb]
\begin{footnotesize}
\centering

\begin{tabular}{|l|c|c|c|c|c|c|c|c|c|c|c|c|}\hline
& $n_{111}$ & $n_{112}$ & $n_{121}$ & $n_{122}$ & $n_{211}$ &
$n_{212}$ &$n_{221}$ &$n_{222}$ &  $n_{311}$ & $n_{312}$
&$n_{321}$ &$n_{322}$ \\ \hline
\hline OL1 &$\ast$ & &$\ast$ & & & & & & & & & \\ \hline OLTCS & &
& $\ast$ & & & & & & & & & \\ \hline
\multirow{2}{*}{OMP} & & & $\ast$ & & & & & & & & & \\
                     & & & & $\ast$ & & & & & & & & \\ \hline
OMPC & $\ast$ & & $\ast$ & $\ast$ & & & & & & & & \\ \hline OMPCL1
& $\ast$ & $\ast$ & $\ast$ & $\ast$ & & & & & $\ast$ & & & \\
\hline

\end{tabular}
\end{footnotesize}
\caption{Identification results for the Social Network example.}
\label{tab:SocialNetwork}
\end{table}

\cite{upton1980} and \cite{upton1995} also analyze the given
contingency table with regard to outliers. They state that
$n_{122}$ should be regarded as an outlier because many pregnant
women are still working and get there by bus. There they see their
co-workers who are also their friends. This cell has been detected
as one of the two solutions of the OMP method, which supports the
hypothesis that it works good for a reasonable model and rather
small contingency tables. The OMPC method also detected $n_{122}$
as an outlier, but not as the only one.

\section{Conclusions}
\label{sec:conclusion}

From the simulations and the real data examples, we can now
summarize the main features of the outlier detection algorithms
considered here.

The OL1 method provides a computationally efficient way to detect
outliers in contingency tables, but the OMPC method in most cases
outperforms this one-step procedure. Using the OMPCL1 method
instead of the OMPC  results in an increase of the outlier
detection rates in most situations while simultaneously decreasing
the inlier detection rates. The OLTCS method can be seen as a
compromise between OL1 and OMPC for medium-sized tables w.r.t.
detection rates, but for bigger tables it is outperformed by the
other procedures. The examples suggest that also the OMP method
works better than the OL1 method.


On the other hand, the detection of outliers becomes difficult
when there are several outliers in one row or in one column (see
the third scenario), and more generally the detection is not easy
when the proportion of outliers with respect to the number of
cells is high, as shown in the last example. However, in practice
we expect to have few outlying cells compared to the dimension of
the table. Finally, when the outlyingness is higher (see the fifth
scenario), the methods identify more outliers as outliers, but
also more inliers as outliers.

Of course, it is worth noting that the experiments performed here
are not exhaustive. Several further simulations should be
implemented to explore the performances of the minimal patterns
algorithms, and to adjust the simulation parameters. In
particular, the behavior of our algorithms for large sparse
tables, or for tables with zero cell counts, still needs to be
explored.

Future work is needed on theoretical results on strictly minimal
patterns for higher dimensional loglinear models to allow for the
development of efficient algorithms. Additionally, alternative
estimation methods might be introduced as well as changes in the
procedure.

Some forward procedures to detect outliers in classical regression
for example also start from minimal subsets of observations
\citep{atkinson2000}, however, without the problem of having to
determine them first. Then the one which minimise the median of
the ordered squared residuals of the remaining observation is
chosen as initial outlier free subset to proceed from. The same
approach could also be followed up for contingency tables.

\section*{Acknowledgements}
FR is partially supported by the Italian Ministry for University
and Research, programme PRIN2009, grant number 2009H8WPX5.

\bibliographystyle{spbasic}      
\bibliography{STCO1417_Revision}

\end{document}